\documentclass[twoside,leqno,11pt]{article}
\usepackage[letterpaper]{geometry}

\usepackage{amsfonts,amsmath,amssymb}
\usepackage{hyperref}
\usepackage[capitalise]{cleveref}
\usepackage{ltexpprt}

\usepackage{microtype} 

\title{Simpler and Faster Contiguous Art Gallery}
\author{Sarita de Berg\thanks{Theoretical Computer Science, IT University of Copenhagen, Denmark} \and Jacobus Conradi\footnote{Theoretical Computer Science, Universität Bonn, Germany} \and Ivor van der Hoog\footnotemark[1] \and Frank Staals\footnote{Department of Computer Science, Utrecht University, the Netherlands}}
\usepackage{xcolor}
\usepackage{graphicx}
\usepackage[shortlabels]{enumitem}
\usepackage{xspace}

\usepackage{caption}
\usepackage[title]{appendix}
\usepackage{pifont}
\usepackage{color}
\usepackage{comment}

\usepackage{lineno}

\graphicspath{ {./figures/} }

\newcommand{\candidateStarts}{\mathcal{U}_P\xspace}
\newcommand{\snapGuards}{\mathcal{G}_P\xspace}
\newcommand{\Gopt}{G^*\xspace}

\newcommand{\eps}{\ensuremath{\varepsilon}\xspace}

\def\contAG{{\textsc{Contiguous Art Gallery}}\xspace}
\newcommand{\ts}{\textsuperscript}
\date{}
\hypersetup{
  colorlinks=true,
  linkcolor=red!70!black,
  linktoc=all,
  citecolor=blue,
    pdfpagemode=FullScreen,
}
\definecolor{Purple}{cmyk}{0.45,0.86,0,0}
\renewcommand{\emph}[1]{{\color{Purple}{\em #1}}}

\begin{document}

\maketitle

\begin{abstract}
The contiguous art gallery problem was introduced at SoCG'25 in a merged paper that combined three simultaneous results, each achieving a polynomial-time algorithm for the problem. This problem is a variant of the classical art gallery problem, first introduced by Klee in 1973. In the contiguous art gallery problem, we are given a polygon~$P$ and asked to determine the minimum number of guards needed, where each guard is assigned a \emph{contiguous} portion of the boundary~$\partial P$ that it can see, such that all assigned portions together cover~$\partial P$. The classical art gallery problem is NP-hard and $\exists \mathbb{R}$-complete, and the three independent works investigated whether this variant admits a polynomial-time solution. Each of these works indeed presented such a solution, with the fastest running in $O(k n^5 \log n)$ time, where $n$ denotes the number of vertices of~$P$ and $k$ is the size of a minimum guard set covering~$\partial P$. We present a solution that is both considerably simpler and significantly faster, yielding a concise and almost entirely self-contained $O(k n^2 \log^2 n)$-time algorithm.
\end{abstract}

\paragraph{Funding.}
This work was supported by the the VILLUM Foundation grant (VIL37507) ``Efficient Recomputations for Changeful Problems''.
Jacobus Conradi is funded by the iBehave Network: Sponsored by the Ministry of Culture and Science of the State of North Rhine-Westphalia and affiliated with the Lamarr Institute for Machine Learning and Artificial Intelligence.

\section{Introduction}

The art gallery problem has enjoyed much attention since its introduction in 1973~\cite{chvatal1975combinatorial}. In its classical form, it asks, given a polygon $P$, for a smallest set of points (or guards) in $P$, such that every point in $P$ is visible to some guard. A point $p$ sees a point $q$, if the line segment $\overline{p\, q}$ is contained in $P$. This problem is known to be NP-hard \cite{lee1986VertexNPhard}, and $\exists\mathbb{R}$-complete~\cite{abrahamsen2018ETRhard}.
Many variants have been considered. This ranges from restricting guards to the edges of the polygon, which was also shown to be $\exists\mathbb{R}$-complete \cite{abrahamsen2018ETRhard}, to restricting guards to the vertices of $P$, which was shown to be NP-hard \cite{lee1986VertexNPhard}. Alternatively, one can restrict the points that must be guarded. If asked to guard only the \emph{boundary} $\partial P$ of the polygon $P$, then the problem is also $\exists\mathbb{R}$-complete \cite{stade2025interiorBoundaryHard}. If the guards are also restricted to the boundary, it is at least NP-hard \cite{lee1986VertexNPhard}. In short: most variants and restrictions of the art gallery problem are hard to solve.

Recently, at SoCG'25~\cite{theEnemy}, the contiguous art gallery problem was introduced. 
In this version of the problem, a \emph{contiguous guard} $g$ is a point in $P$ that is assigned a contiguous visible portion  $[u, v]$ of  $\partial P$. 
The \contAG problem asks for a given polygon $P$ to compute a smallest set of contiguous guards where, together, they see all of $\partial P$. 
Three works show that this problem variant is polynomially solvable. Let $P$ have $n$ vertices and let the optimal solution have $k$ guards. 
Merrild, Rysgaard, Schou, and Svenning~\cite{merrild2024contiguous} show an $O(k n^5 \log n)$-time algorithm via a greedy procedure.
Starting from an arbitrary point $u$ on $\partial P$, they compute the maximal portion $[u, v]$ along $\partial P$ such that there exists a contiguous guard $(g, [u, v])$, and recurse on $v$. 
To compute such maximal portions, they show an involved data structure that computes the visibility polygons, which is the area of $P$ seen by a given point, of all vertices in $P$ and all points along $\partial P$ that their algorithm encounters.  
This procedure returns a solution of at most $k + 1$ contiguous guards.
They restart their procedure with a starting point that is not already in any of the previous solutions and after $O(kn^3)$ restarts they find solution of size $k$ (using $O(k n^5 \log n)$ total time).
Biniaz, Bose, Maheshwari, Mitchell, Odak, Polishchuk and Shermer~\cite{biniaz2024contiguous} compute $O(n^4)$ possible guard locations. From this set, they apply an intricate combinatorial argument to combine these candidate contiguous guards into the optimal solution in $O(k n^5 \log n)$ time. 
Robson, Spalding‑Jamieson, and Zheng~\cite{robson2024analytic} reduce the problem to an abstract problem that they call the Analytic Arc Cover Problem (AACP). Using results from real algebraic geometry and quantifier elimination, they show that AACP (and thereby \contAG) is polynomially solvable.

We revisit the \contAG problem.
We give a simple but strong structural theorem that defines a set of $O(n^2)$ starting locations $\candidateStarts$, such that for at least one $u \in \candidateStarts$, a greedy algorithm that starts at $u$ finds an optimal solution. 
This in itself already improves the running time of~\cite{merrild2024contiguous} by a factor $\Theta(kn)$~\cite[Theorem 3]{theEnemy}. However, we additionally show a data structure that for any point $u \in \partial P$, computes the maximal portion $[u, v]$ along $\partial P$ such that there exists a contiguous guard $(g, [u, v])$ in $O(\log^2 n)$ time.
This results in an $O(k n^2 \log^2 n)$-time algorithm, which is significantly more efficient. 
We argue that our approach is also considerably simpler, as we rely only upon basic geometric observations and techniques such as shortest paths in simple polygons, and intersection algorithms for convex polygons.
Our result is inarguably shorter, as the previous works require 48~\cite{merrild2024contiguous}, 22~\cite{biniaz2024contiguous} and 21~\cite{robson2024analytic} pages.


\section{Preliminaries}
\label[section]{sec:Preliminaries}

We denote by $P$ the area bounded by a curve that forms a simple polygon of $n$ vertices defining $n$ edges. 
We denote by $\partial P$ its boundary and we assume that the sequence of vertices of $\partial P$ is given in counterclockwise order. We see each edge as an ordered pair $\overline{u \, v}$ using this counterclockwise ordering. This means the edge is directed from $u$ to $v$.
A vertex of $P$ is a \emph{reflex vertex} if its interior angle is more than $\pi$. 
Since consecutive vertices $v_i$ and $v_{i+1}$ of $P$ are in counterclockwise order, the interior of $P$ lies immediately left of the edge from $v_i$ to $v_{i+1}$. 
With this terminology, we define three core concepts.


A \emph{chain} is a collection of edges defined by a sequence of vertices. 
For two points on the boundary $u$, $v$ we denote by $[u, v]$ the chain that is formed by traversing $\partial P$ from $u$ to $v$ counterclockwise.
We denote by $(u, v)$ the open chain, which is the set of points $(u, v) := \{  x \in [u, v] \mid x \neq u, x \neq v \}$. 
For a fixed point $u\in\partial P$, and points $v_1, v_2\in \partial P$ we say that $v_1 \leq v_2$ if $[u, v_1] \subseteq [u, v_2]$.

A \emph{contiguous guard} is a point $g\in P$ together with a chain $[u,v]$ such that any point in the chain is seen by $g$.
For a fixed point $u\in\partial P$, we frequently wish to compute a \emph{maximal} point $v\in\partial P$ such that there exists a point $g \in P$ such that $(g, [u, v])$ is a contiguous guard.
By this we mean that for all other contiguous guards $(g', [u, v'])$, $v' \leq v$ (i.e., $[u, v'] \subseteq [u, v]$).  

For any set of (directed) edges $E$, we define the \emph{visibility core} $P_E$ as the set of points $p\in \mathbb{R}^2$ that lie left of every (directed) supporting line (so $p$ lies left of, or on the line) of every edge in $E$. In particular, we can use the visibility core to decide whether a polygon can be guarded by a single contiguous guard.

\begin{observation}\label{obs:leftOfEdges}
    If $(g,[u,v])$ is a contiguous guard, then $g$ lies left of the supporting lines of all edges of $P$ intersecting the open chain $(u,v)$, and in particular in the visibility core of these edges.
\end{observation}

\begin{observation}\label{obs:trivialOneCenter}
    There exists a contiguous guard that can see all of $\partial P$ if and only if the visibility core of all edges of $P$ is not empty (which can be verified in linear time).
\end{observation}

Henceforth, we assume that $P$ cannot be guarded with a single contiguous guard.


\section{Technical overview and contribution}
We give an algorithm for the \contAG problem with running time $O(kn^2\log^2 n)$, where $k$ is the smallest number of contiguous guards required to guard all of $\partial P$. Similar to the approach of the best known algorithm for this problem---which runs in $O(kn^5\log n)$ time \cite{theEnemy, merrild2024contiguous}---we base our algorithm on the following observation. Consider a greedy algorithm that proceeds as follows: Given start point $u\in \partial P$, compute the maximal $v\in \partial P$ such that there exists a contiguous guard $(g, [u, v])$ and add this guard to the solution set. Then recurse on $v$ until the guards cover $\partial P$. 
Invoking this greedy algorithm for an arbitrary $u \in \partial P$  yields a solution of size at most $k+1$.
The algorithm from~\cite{theEnemy, merrild2024contiguous}  roughly states that ``after $O(kn^3)$ repeated applications of the subroutine, you find a solution of size~$k$''.

We show that if instead the greedy algorithm starts from a position $u$ where there exists an optimal solution $G$ with  $(g, [u, v]) \in G$, then it outputs an optimal solution. 
We compute $O(n^2)$ candidate starting locations for $u$, and we prove that at least one of these is guaranteed to produce a solution of size~$k$. This yields an immediate runtime improvement to $O(n^4\log n)$, via the greedy subroutine of~\cite{theEnemy, merrild2024contiguous}. However, by using our candidate set of $O(n^2)$ points and a shortest path data structure, we provide a new subroutine which takes $O(\log^2n)$ time, which immediately yields an $O(k \, n^2\log^2 n)$-time algorithm.

To obtain our principal structural insight---the $O(n^2)$ candidate starting locations---we consider the arrangement $C_P$ formed by the supporting lines of edges of $P$. 
We show that any guarding solution $G$ can be transformed into a guarding solution $G'$ of equal size such that for at least one guard $(g, [u, v]) \in G'$, the point $g$ is a vertex of $C_P$. 
Furthermore, the chain $[u, v]$  is inclusion-wise maximal and intersects the two edges whose supporting lines define $g$. This yields our set of $O(n^2)$ candidate chains by computing for all vertices  $c$ of $C_P$ the inclusion-maximal chains $[u, v]$ that intersect an edge defining $c$. 

Our key observation is the fact that we can split the set of all possible guards into ``good'' and ``bad'' guards. Conceptually, a contiguous guard $(g, [u, v])$ is good, if the angle $\sphericalangle(u,g,v)$ is at most $\pi$, while for bad guards the angle $\sphericalangle(u,g,v)$ is more than $\pi$. We show that any good guard $(g, [u, v])$ can be replaced by a guard $(c, [u', v'])$ where $c$ is a vertex of $C_P$ and $[u, v] \subseteq [u', v']$. We further show that any solution with more than one guard can be transformed to contain at least one good guard. 

Given this structural insight, we provide a faster subroutine for our greedy algorithm.
In particular, we use shortest path and ray shooting data structures to compute for any point $u$, the maximal $v\in \partial P$ such that there exists a contiguous guard $(g, [u, v])$ (our structure can also output $(g, [u, v])$). 
For this data structure, our structural insight is again crucial. We prove that if $(g, [u, v])$ is a good guard then there exists a chain $[u', v']$ in our set of $O(n^2)$ candidate chains such that $u \in [u', v']$ and $v'$ is maximal (with respect to $u$). 
By storing all candidate chains (and their guards) in a sorted array, we obtain $[u', v']$ and its guard $g'$ in logarithmic time. If $(g, [u, v])$ is not in our set of candidate chains, then the fact that it must be a bad guard provides us with a lot of information about the structure of $(g, [u, v])$ which we use to compute it in $O(\log^2 n)$ time via precomputed visibility cores and data structures.




\section{Structuring guard locations}\label{sec:goodBad}

Consider all supporting lines of edges in $P$ and their arrangement $C_P$.
Our key insight shows that a contiguous guard $(g, [u, v])$ is ``well-behaved'', \emph{unless} it  meets four very specific conditions.
Formally:

\begin{definition}
    \label{def:bad}
    We say that a contiguous guard $(g, [u, v])$ is \emph{bad} if it meets all the following conditions:
    \begin{enumerate}[label=\textup{(\Roman*)}, ref=\textup{(\Roman*)}, noitemsep,nolistsep]
              \item \label{itm:first} the open chain $(u,v)$ contains a vertex of $P$, 
        \item \label{itm:second} $u \neq g$ and $v \neq g$, 
        \item \label{itm:third} the angle $\sphericalangle(v,g,u) > \pi$, and
        \item \label{itm:fourth} the shortest path from $u$ to $v$ is not a single edge. 
    \end{enumerate}
    We say the guard is \emph{good} if it is not bad.
\end{definition}

We prove that for all good guards $(g, [u, v])$ there exists a vertex $c$ of $C_P$ such that $(c, [u, v])$ is also a contiguous guard. 
Every vertex in $C_P$ is defined by two edges. 
We furthermore prove that, unless $[u, v]$ contains no vertex of $P$, $(u, v)$ contains at least one vertex of both edges defining $c$. 

\subsection{Replacing good guards.}
A contiguous guard is \emph{good} if it fails any of the four conditions of Defintion~\ref{def:bad}. 
We consider a contiguous guard $(g, [u, v])$. We denote by $E$ all edges intersected by the open chain $(u, v)$ and by $P_E$ the visibility core of $E$. 
For all $i \in \{\textup{(I)},\ldots, \textup{(IV)}\}$
we suppose that $(g, [u, v])$ is \emph{good} because it meets all conditions $j$ for $j < i$ but not condition $i$. 
We prove, in Lemmas~\ref{lem:first},~\ref{lem:second}, \ref{lem:third}, and \ref{lem:fourth}, for these four cases that there exists a contiguous guard $(c, [u, v])$ where $c$ is a vertex of $P_E$ or a vertex of $E$ (thus, $c$ is defined by the intersection of supporting lines of two edges in $E$).

\begin{lemma}[Failing \cref{itm:first}]
    \label{lem:first}
    Let $(g,[u,v])$ be a contiguous guard such that the open chain $(u,v)$ does not contain a vertex of $P$. If $e$ is the edge on which $(u, v)$ lies then either vertex of $e$ also sees $(u, v)$.
\end{lemma}

\begin{proof}
    As $e$ is a line segment contained in $P$, 
    the endpoints of $e$ see all points on the edge.
\end{proof}

For the remainder, we assume that for $(g, [u, v])$, $[u, v]$ contains at least one vertex. 
Our proofs of Lemmas~\ref{lem:second}--\ref{lem:fourth} rely on the following proposition:

\begin{proposition}\label{prop:visibilityCoreVertex}
    Let $(g,[u,v])$ be a contiguous guard such that $[u,v]$ contains a vertex of $P$. Let $E$ be the edges of $P$ intersecting $(u, v)$ and $P_E$ be its visibility core. Denote by $P_g$ the polygon defined by $[u,v]$, $\overline{v \, g}$ and $\overline{g \, u}$. Then either a vertex of $P_E$ lies in $P_g$, or, a vertex of $P$ in $[u,v]$ lies in $P_E$.
\end{proposition}
\begin{proof}    
    Figure~\ref{fig:180snap} illustrates the objects in the lemma statement. Observe that $\partial P_E$ forms a polygonal chain $C$. Let $e_u$ and $e_v$ be the edges of $E$ containing $u$ and $v$, respectively. The chain $C$ lies left of the supporting lines of $e_u$ and $e_v$. As $g\in P_E$, the chain $C$ intersects the edges $\overline{g\,u}$ and $\overline{v\,g}$. Let $C'$ be the
    subchain of $C$ contained in $P_g$. 
    We claim that either a vertex of $P_E$ or a vertex of $P$ in $[u,v]$ lies on $C'$.

    For the sake of contradiction, assume that no point of $C'$ is a vertex of $P_E$. Then $C'$ is defined by a single edge of $P_E$ with its endpoints outside of $P_g$. This edge is defined by a supporting line of an edge $e \in E$ where at least one endpoint of $e$ lies in $[u,v]$. This endpoint is a vertex of $P$ in $[u,v]$, and it lies in $C'$, and thus $P_E$, concluding the proof.
\end{proof}

\begin{figure}
    \centering
    \includegraphics{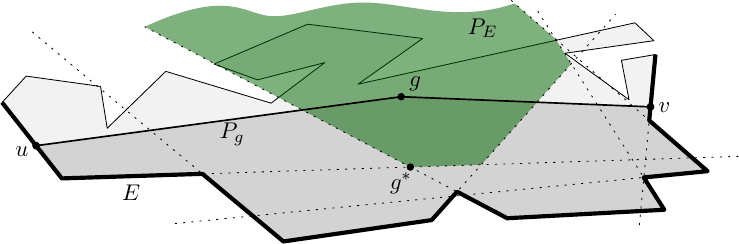}
    \caption{Illustration of the visibility core $P_E$ (in green) of the set of edge $E$ of $P$ intersecting $(u,v)$ (in black bold).
    Any vertex $g^*$ of $P_E$ in $P_g$ is a point that can see the entire dark grey polygon \vspace{-0.3cm}}
    \label{fig:180snap}
\end{figure}


\begin{lemma}[Failing \cref{itm:second}]
    \label{lem:second}
    Let $(g,[u,v])$ be a contiguous guard such that \textup{(I)} $[u,v]$ contains a vertex of $P$, but \textup{(II)} $u = g$ (the case $v = g$ is symmetric).
    Let $E$ be the edges of $P$ intersecting $(u, v)$ and $P_E$ be its visibility core. Then either a vertex of $P_E$, or a vertex of $P$ in $[u,v]$,  also sees $[u,v]$. 
\end{lemma}
\begin{proof}
    Let $P_g$ be the polygon defined by $[u,v]$ and $\overline{v\, u}$. Note that $P_g$ is a simple polygon contained in $P$ because $u = g$ can see $v$.  
    Proposition~\ref{prop:visibilityCoreVertex} implies that either $P_E$ has a vertex in $P_g$, or a vertex of $P$ in $[u,v]$ is in $P_E$. Now, as the entirety of $P_g$ lies left of the supporting line of the edge $\overline{v\, u}$, so does this vertex. Hence, Observation~\ref{obs:trivialOneCenter} implies, this vertex sees all of $\partial P_g$ and in particular $[u,v]$.
\end{proof}

\begin{lemma}[Failing \cref{itm:third}]\label{lem:third}
    Let $(g, [u, v])$ be a contiguous guard such that \textup{(I)} $(u,v)$ contains a vertex of $P$, and \textup{(II)} $u\neq g \neq v$, but \textup{(III)} the angle $\sphericalangle(v,g,u) \leq \pi$.
    Let $E$ be the edges of $P$ intersecting $(u, v)$ and $P_E$ be its visibility core. Either a vertex of $P_E$, or a vertex in $[u,v]$, also sees $[u,v]$.   
\end{lemma}
\begin{proof}
    Refer to Figure~\ref{fig:180snap}.
    Let $P_g$ be the simple polygon defined by $\overline{ g \, u }$ and $\overline{ v \, g}$, and the chain $[u,v]$. Then $P_g$ is contained in $P$ because $g$ can see both $u$ and $v$. By Proposition~\ref{prop:visibilityCoreVertex}, either $P_E$ has a vertex in $P_g$, or a vertex of $P$ in $[u,v]$ is in $P_E$. As the entirety of $P_g$ lies left of the supporting lines of the edges $\overline{ g \, u }$ and $\overline{ v \, g}$, so does this vertex.  
    Hence, Observation~\ref{obs:trivialOneCenter} implies this vertex sees all of $\partial P_g$.
\end{proof}


\begin{lemma}[Failing \cref{itm:fourth}]
    \label{lem:fourth}
    Let $(g,[u,v])$ be a contiguous guard such that \textup{(I)} the open chain $(u,v)$ contains a vertex of $P$, \textup{(II)} $u\neq g\neq v$, and \textup{(III)} the angle $\sphericalangle(v,g,u)>\pi$, but \textup{(IV)} the shortest path from $u$ to $v$ in $P$ is a single edge.
 Let $E$ be the edges of $P$ intersecting $(u, v)$ and $P_E$ be its visibility core. Then either a vertex of $P_E$, or a vertex of $P$ in $[u,v]$, also sees $[u,v]$.  
\end{lemma}
\begin{proof} 
    Let $P'$ be the polygon defined by $[u,v]$ and $\overline{vu}$, and let $P_g$ be the polygon defined by $[u,v]$, $\overline{v \, g}$ and $\overline{g \, u}$. Then $P_g\subseteq P '\subseteq P$. 
    Let $E'$ be the set of edges of $P'$ and denote by $P_{E'}$ its visibility core. By definition of the visibility core, $P_{E'}=P_E\cap P'$. 
    By Proposition~\ref{prop:visibilityCoreVertex}, either a vertex $x$ of $P_E$ is in $P_g$, or a vertex $x'$ of $P$ in $[u, v]$ is in $P_E$. 
    Since $P_g \subseteq P'$ and $P_{E'} = P_E\cap P'$, the set $(P_E \cap P_g) \subseteq P' \cap P_{E'}$. 
    Thus, by Observation~\ref{obs:trivialOneCenter},  the point $x$ (or $x'$) sees all of $P'$ and in particular $[u,v]$.   
\end{proof}

\begin{corollary}
\label{cor:good_is_good}
For any good guard $(g, [u, v])$ there exists a contiguous guard $(c, [u, v])$ where $c$ is a vertex defined by the supporting lines of edges  $e_1$ and $e_2$ of $P$.  
Moreover, unless $[u, v]$ contains no vertex of $P$, we can choose $c$ such that $[u, v]$ contains a vertex of each of $e_1$ and $e_2$. 
\end{corollary}

\begin{proof}
    The first claim follows directly from Lemmas~\ref{lem:first},~\ref{lem:second}, \ref{lem:third}, and \ref{lem:fourth}. 
    For the second claim, observe that Lemmas~\ref{lem:second}, \ref{lem:third}, and \ref{lem:fourth} places $c$ either on a vertex of $P$ within $[u, v]$, or, on a vertex of the visibility core $P_E$ where $E$ is the set of all the edges of $P$ intersected by $(u, v)$. 
\end{proof}

\subsection{Replacing bad guards.}
For all good guards $(g, [u, v])$ we can find a vertex $c$ of the arrangement $C_P$ such that $(c, [u, v])$ is a contiguous guard. 
We cannot guarantee this property for bad guards. However, we \emph{can} find an alternative contiguous guard $(c, [u, v])$ that is somewhat well-behaved. 

\begin{figure}
    \centering
    \includegraphics{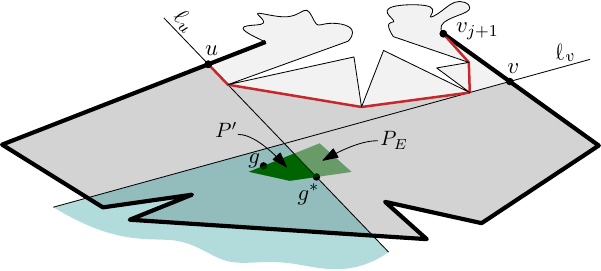}
    \caption{Illustration of proof of Lemma~\ref{lem:AlmostSnapBadGuards} with $P_v$ in gray and the shortest path from $v_{j+1}$ to $u$ in red.}
    \label{fig:badLemma}
\end{figure}

\begin{observation}[See for example visibility polygons in~\cite{Guibas1086Linear}]\label{obs:straight_is_convex}
    Let $g, u, v \in P$ such that $g$ can see $u$ and $v$ and $\sphericalangle(v,g,u) > \pi$, then the shortest path from $u$ to $v$ in $P$ is a convex left-turning chain. 
\end{observation}

\begin{lemma}\label{lem:AlmostSnapBadGuards}
    Let $(g, [u, v])$ be a contiguous bad guard. 
    Furthermore, let $v$ lie on edge $e = \overline{v_j \, v_{j+1}}$ of $P$. 
    Consider the shortest path from $v_{j+1}$ to $u$ and let $\ell_u$ be the supporting line of the last edge on this path. 
    If $E$ is the set of edges of $P$ intersecting $(u, v)$ and $P_E$ is its visibility core, then either a vertex of $P_E$, or, the first or the last point along $\ell_u$ in $P_E$,  also sees $[u,v]$.    
\end{lemma}

\begin{proof}
    We note that this proof does not use the fact that the open chain $(u,v)$ contains a vertex of $P$, nor that the angle $\sphericalangle(v,g,u)>\pi$. These are used later in the paper. We refer to \Cref{fig:badLemma}.
    
    Let $S_v$ be the shortest path from $v$ to $u$ which, per definition, is not a single edge. 
      Let $\ell_v$ be the supporting line of the first edge of $S_v$ (the edge starting at $v$).
    Observe that $\ell_u$ is the supporting line of the last edge of $S_v$ (the edge ending at $u)$. 
    Since $g$ can see both $u$ and $v$, the edges of $S_v$ form a left-turning chain (Observation~\ref{obs:straight_is_convex}). 
    It follows that the visibility core of $S_v$ is the intersection of the halfplanes left of $\ell_u$ and $\ell_v$. 
   Observe that the polygon $P_v$ defined by $S_v$ and the chain $[u, v]$ is well-defined (as $g \neq u$ and $g \neq v$).
    The visibility core $P'$ of $P_v$ is the intersection of $P_E$ and the visibility core of $S_v$, and $P'$ thus equals equals $P_E$ intersected by the halfplanes left of $\ell_u$ and $\ell_v$. 
    
    The guard $g$ lies in the visibility core of $S_v$ and $P_E$, and so $P'$ is not empty. 
    If at least one point of intersection $g^*$ between $\ell_u$ and $P_E$ is in $P'$, then by Observation~\ref{obs:trivialOneCenter}, the point $g^*$ can see all of $P_v$ (and in particular $[u,v]$).
    Otherwise, we note that $P'$ is non-empty and bounded in $P$, so $P'$ contains a vertex of $P_E$ or $P$ and  we conclude the lemma.
\end{proof}

\noindent
In the remainder of this paper we show three things:
\begin{itemize}[noitemsep,nolistsep]
    \item We first use our good guards to discretize the boundary of $P$ into a vertex set $\candidateStarts$ of complexity $O(n^2)$ such that there exists an optimal solution $G$ where at least one $(g, [u, v]) \in G$ has $u \in \candidateStarts$. 
    \item We next use our categorization of both bad and good guards to construct a data structure where, for any start point $u$, we can compute the guard $(g, [u, v])$ that maximizes $v$ in $O(\log^2 n)$ time. 
    \item Finally, we present an $O(k n^2\log^2 n)$-time algorithm that for every vertex  $u \in \candidateStarts$ starts a greedy procedure that iteratively adds the guard $(g, [u, v])$ that maximizes $v$ until $\partial P$ is guarded.  
\end{itemize}

\section{Discretization of optimal solutions}\label{sec:goodStarts}
We compute a set $\candidateStarts$ of $O(n^2)$ starting points such that there exists at least one $u \in \candidateStarts$ and at least one optimal guarding solution $G$ with $(g, [u, v]) \in G$. Observe that for a fixed point $c$, its visibility polygon and therefore the inclusion-wise maximal chains along $\partial P$ that $c$ can see are also fixed.

\begin{definition}
    \label{def:guard_set}
    Let $P$ be a simple polygon and denote by $C_P$ the arrangement of all supporting lines of edges in $P$.
    We define the \emph{candidate set} $\snapGuards$ as follows: 
    For each vertex $c \in C_P$ let $\overline{a\,b}$ and $\overline{e\,f}$ be its defining edges. 
    For each $x \in \{a, b, e, f\}$ let $(c, [u_x, v_x])$ be a contiguous guard such that $[u_x, v_x]$ is the unique inclusion-wise maximal interval containing $x$. 
    If such a guard 
    $(c, [u_x, v_x])$ exists, we add it to~$\snapGuards$.    
\end{definition}

We can use $\snapGuards$ to discretize the boundary of $P$:

\begin{definition}
    \label{def:V^*}
    Let $P$ be a simple polygon and $\snapGuards$ be its set of discrete guards. We create a set of points on $\partial P$  denoted by $\candidateStarts$ which consists of all vertices in $P$, and $u_x$ for every $(c, [u_x, v_x]) \in \snapGuards$. 
\end{definition}


\begin{lemma}\label{lem:verticesAreGood}
    Let $(g, [u,v])$ be a contiguous guard such that $[u,v]\neq \partial P$ is maximal. Then either $v$ is a reflex vertex of $P$ or $\overline{v \,g}$ contains a reflex vertex of $P$ in its interior. 
\end{lemma}
\begin{proof}
        Since $[u, v]$ is maximal, there exists a value $\eps^* > 0$ such that for all $\eps \in (0, \eps^*]$, the shortest path from $g$ to the point $v + \eps$ on $\partial P$ visits a reflex vertex $x$ of $\partial P$. 
        If $x = v$ then $v$ is a reflex vertex of $P$. 
        Otherwise, $\overline{v \, g}$ contains $x$ in its interior. 
\end{proof}

\begin{figure}
    \centering
    \includegraphics{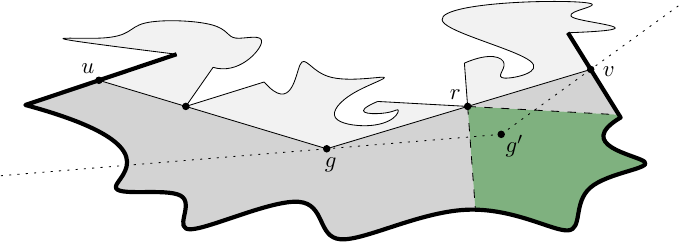}
    \caption{Illustration of Theorem~\ref{thm:good_seed} which shows a bad contiguous guard $(g, [u, v])$. For a description of the vertices in the figure, we refer to the corresponding proof. We show the polygon $P_g$ in gray.}
    \label{fig:goodSeed}
\end{figure}

\begin{theorem}\label{thm:good_seed}
      Let $G$ be a minimal set of contiguous guards guarding $\partial P$ with $|G| > 1$. There exists a set of contiguous guards $G'$ with $|G| = |G'|$ where at least one $(g, [u, v]) \in G'$ has  $u \in \candidateStarts$.
\end{theorem}

\begin{proof}
    We show how to transform $G$ into such a set $G'$.
    We start by replacing every guard $(g, [u, v]) \in G$ by a guard $(g, [u', v'])$ such that $[u, v] \subseteq [u', v']$ and the interval $[u', v']$ is maximal, we continue as follows.

    Suppose that $G$ contains a good guard $(g, [u, v]) \not \in \snapGuards$. 
       If $[u, v]$ contains no vertex of $P$ then by Lemma~\ref{lem:first}, we may replace $(g, [u, v])$ by $(a, [a, b])$ where $\overline{ab}$ is the edge containing $[u, v]$ and we conclude the theorem.
       Otherwise, we apply Corollary~\ref{cor:good_is_good} and 
    note that there exists a guard $(c, [u, v])$ where $c$ is a vertex of $C_P$ and $[u, v]$ contains a vertex $x$ of the defining edges of $c$.
    We observe that $\snapGuards$ contains a guard $(c, [u_x, v_x])$ where $[u_x, v_x]$ contains $x$ and is maximal. So, $[u, v] \subseteq [u_x, v_x]$. We update $G$ by replacing $(g, [u, v])$ by $(c, [u_x, v_x])$, concluding the proof. 
    
    Now, suppose that $G$ contains only bad guards and let $(g, [u, v]) \in G$ be such a bad guard. 
    If $u$ is a vertex of $P$ then we conclude the proof.
    If $v$ is a vertex of $P$ then, because $|G| > 1$, there has to be at least one other guard $(g', [u', v']) \in G$ such that $v \in [u', v']$. We replace this guard by $(g', [v, v'])$ so that the start of its chain is a vertex of $P$ and we conclude the proof. 
    
    If neither $u$ nor $v$ is a vertex of $P$, then we apply Lemma~\ref{lem:verticesAreGood} to note that $\overline{v \, g}$ contains a reflex vertex $r$ of $P$ with $r \not \in [u, v]$. Thus, there exists some guard $(g', [u', v'])$ with $r \in [u', v']$. 
    If $u' = r$ or $v' = r$ then we conclude the proof as above. 
    If instead $r \in (u', v')$, then $g'$ lies left of both supporting lines of edges incident to $r$. This in turn implies that $g'$ is left of the supporting line of $\overline{v\,g}$. As $\sphericalangle(v,g,u)>\pi$, any point that sees $r$ and is left of the supporting line of $\overline{v\,g}$ must 
    lie in the polygon $P_g$ defined by $[u, v]$, $\overline{v \, g}$ and $\overline{g \, u}$ (see \Cref{fig:goodSeed}).
    Observe that the points in $[u', v'] \setminus (u, v)$ form a single closed chain, else $[u, v] \subseteq [u', v']$, which contradicts the minimality of $G$. 
    Define $[a, b] = [u', v'] \setminus (u, v)$. 
    If we replace $(g', [u', v'])$ by $(g', [a, b])$ the resulting set of contiguous guards still guard $\partial P$. However, crucially, $[a,b]$ lies left of the supporting line through $\overline{g' \, v }$ and left of either the supporting line of $\overline{u \, g'}$ or $\overline{g \, g'}$. 
    As $g'$ lies in $P_g$ and $\sphericalangle(v,g,u)>\pi$. this implies that $\sphericalangle(a,g',b) \leq \pi$, but then $G$ contains a good guard.
\end{proof}

\section{A guarding data structure}\label{sec:visibilityCompute}

Definition~\ref{def:guard_set} defines a discrete set of guards $\snapGuards$. 
We prove that we can compute $\snapGuards$ in $O(n^2 \log n)$ time. 
Our strategy is as follows: let $C_P$ be the arrangement of lines supporting edges of $P$. 
For each vertex $c \in C_P$ with defining edges $\overline{a\,b}$ and $\overline{e \, f }$ we want to compute for $x \in \{ a, b, e, f\}$ the unique inclusion-wise maximal chain $[u, v]$ such that $(c, [u, v])$ is a contiguous guard and $x \in [u, v]$. 
We note that for all edges of $P$ that intersect $(u, v)$, the point $c$ must lie left of the supporting line of the edge. 
For $x \in \{ a, b, e, f \}$, we first consider maximal contiguous sequences of edges that have this property:

\begin{definition}
    \label{def:interval}
    Let $g$ be a point in $P$ and let a point $m \in \partial P$ be given such that $g$ sees $m$.
    We define $E_{g,m}=\{e_i,e_{i+1},\ldots,e_j\}$ as the maximal contiguous sequence of edges along $\partial P$ such that $m$ lies on an edge of $E_{g,m}$, and $g$ lies left of the supporting line through each edge of $E_{g,m}$, i.e., in its visibility core.
\end{definition}

\begin{figure}
    \centering
    \includegraphics{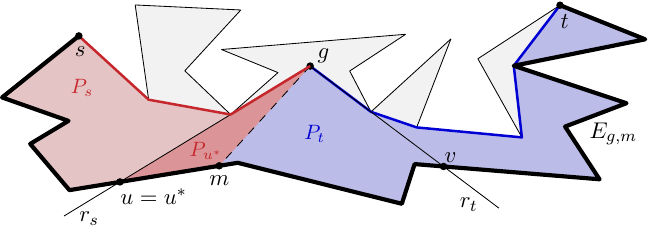}
    \caption{Computing the inclusion-wise maximal area that can be seen by $g$ containing $m$ via the contiguous edge set $E_{g,m}$ (bold black edges) of which $g$ lies left of, defining $s$ and $t$. With the shortest paths from $s$ (in red) and $t$ (in blue) to $g$ respectively, the rays $r_s$ and $r_t$ are constructed. The intersection with $[s,t]$ form $u$ and $v$. The interval $[u,v]$ is inclusionwise maximal, s.t.~it is visible from $g$ and $m\in[u,v]$.}
    \label{fig:visiblityCompute}
\end{figure}

\begin{lemma}\label{lem:computeVisibility}
    Let $g$ be a point in $P$ and a point $m \in \partial P$ be given such that $g$ sees $m$. 
    Let $s$ and $t$ be the vertices defining the start and end of a contiguous subset $E\subseteq E_{g,m}$, such that $m\in[s,t]$ and $[s,t]\neq\partial P$. 
    Consider the shortest path $\pi_s$ from $s$ to $g$ in $P$ and define by $r_s$ the ray from $g$ through the last edge of $\pi_s$. 
    We define $\ell_t$ analogously.
    Then the rays  $r_s$ and $r_t$ hit $[s,t]$ in $u$ and $v$ where $[u, v]$ defines the unique maximal chain along $\partial P$ that is visible from $g$, contains $m$ and is contained in $[s,t]$ (i.e., in $E$). 
\end{lemma}

\begin{proof}
    Consider the polygon defined by $\pi_s$, $[s, t]$ and the shortest path $\pi_t$ from $t$ to $g$. As $g$ sees $m$, the segment $\overline{g\,m}$ splits this polygon into two polygons $P_s\subset P$ and $P_t\subset P$ ($P_s$ contains $s$, and $P_t$ contains~$t$). 
    As $[s,t]\neq\partial P$, $P_s$ and $P_t$ only intersect in $\overline{g\, m}$ and in particular, the shortest paths $\pi_s$ and $\pi_t$ are disjoint except for possibly at their last edge (ending at $g$), which may be colinear with $\overline{g\,m}$.

    Let us now consider the maximal $u^*\in[s,m]$ such that $\overline{g\,u^*}$ intersects the shortest path from $g$ to $s$ in more than one point (i.e. not only in $g$).
    First, we show that $\overline{g\,u^*}\subset P$.
    The point $g$ lies left of all supporting lines of edges of the polygon $P_{u^*}$ defined by $[u^*,m]$, $\overline{m\,g}$, and $\overline{g\,u^*}$, and hence by Observation~\ref{obs:trivialOneCenter}, $g$ sees all of $[u^*, m]$ in $P_{u^*}$. By definition of $u^*$, $P_{u^*}$ contains no point of $\pi_s$ in its interior and $P_{u^*}$ is thus contained in $P_s\subset P$. Hence, $g$ sees $[u^*, m]$ in $P$ and $\overline{g\,u^*}\subset P$. 
    
    Next, observe that as $\overline{g\,u^*}\subset P$, $u^*$ is the point in $E$ that is hit by the ray $r^*$ from $g$ to $u^*$. We claim that $r^* = r_s$. Indeed, $\overline{g\,u^*}$ intersects $\pi_s$. So, if   $\overline{g\,u^*}$ intersects $\pi_s$ in a point $r$ not on the last edge of $\pi_s$, we could shorten $\pi_s$ by the segment $\overline{r\,g}\subset \overline{u^*\,g}\subset P$.
    Finally, observe that if $s\neq u^*$, then $g$ does not see any point $u'$ in $[s,u^*)$, as by definition of $u^*$ the segment $\overline{g\,u'}$ intersects an edge of $\pi_s$ in its interior.
    
    Lastly, for the points $u=u^*$ and $v=v^*$ along $E$ that are hit by $r_s$ and $r_t$, we already observed that $g$ sees $[u,m]$ and $[m,v]$ and hence $[u,v]$ is as claimed.
\end{proof}



\noindent
We can now compute $\snapGuards$ using the following well-known data structures for simple polygons~\cite{Guibas1987Optimal, Hershberger1993Ray}:

\begin{lemma}[\cite{Guibas1987Optimal}]\label{lem:shortestPath}
    We can store a polygon $P$ with $n$ vertices using $O(n)$ time and space, such that given any $s, t \in P$ we can output the shortest path from $s$ to $t$ as $O(\log n)$ balanced trees in $O(\log n)$ time.
\end{lemma}

\begin{lemma}[\cite{Hershberger1993Ray}]\label{lem:rayShooting}
     We can store a polygon $P$ with $n$ vertices using $O(n)$ time and space, such that given any ray $r$ with its origin in $P$ we can detect the point on $\partial P$ hit by $r$ in $O(\log n)$ time. 
\end{lemma}

\begin{lemma}\label{lem:realComputePrimitive}
    Let $g$ be a point in $P$ and let $m\in\partial P$ be given. Let $s$ and $t$ be the vertices defining the start and end of a contiguous subset $E\subseteq E_{g,m}$, such that $m\in[s,t]$ and $[s,t]\neq\partial P$. Moreover, let $g \neq s$ and $g \neq t$. Given the data structures from Lemmas~\ref{lem:shortestPath} and \ref{lem:rayShooting}, one can compute the maximal chain $[u, v]$ along $\partial P$ that is visible from $g$, contains $m$ and is contained in $[s,t]$ (i.e., in $E$) in $O(\log n)$ time. 
\end{lemma}
\begin{proof}
    This is an immediate consequence of Lemma~\ref{lem:computeVisibility}, where $u$ and $v$ are obtained by with ray shooting from $g$ along the last edges of the shortest paths from $s$ to $g$ and $t$ to $g$. 
\end{proof}

\begin{lemma}\label{lem:computeAllVisibility}
    We can compute the set $\snapGuards$ in $O(n^2 \log n)$ time and $O(n^2)$ space. 
\end{lemma}
\begin{proof}
    First, we construct the data structures from Lemma~\ref{lem:shortestPath} and \ref{lem:rayShooting}.  
    For all supporting lines of edges in $P$, we compute their arrangement $C_P$ inside of $P$. 
    If the vertices of $P$ lie in general position, each vertex on this arrangement lies on exactly two supporting lines. 
    We define for each vertex $c$ of this arrangement a bit string $\{0,1\}^n$ where the $i$\ts{th} bit is a $1$ if and only if  $c$ lies to the left of the supporting line of the $i$\ts{th} edge of $P$.  
    We store this bit string in a balanced binary tree where the $i$\ts{th} leaf stores the corresponding bit and each inner node stores whether its subtree contains a leaf with at least one $0$-bit.
    
    By virtue of the arrangement being defined by all lines supporting edges of $P$, this bit string differs by at most two entries for neighboring vertices in the arrangement.   
    We now compute an arbitrary walk of length $O(n^2)$ through the planar graph which passes every vertex of $C_P$. If we walk from a vertex $g$ to a vertex $g'$ in this arrangement, we leave exactly one supporting line $\ell$ and arrive on exactly one other supporting line $\ell'$.
    By comparing $g'$ to both these lines, we can update the bit string to correspond to the bit string of $g'$. 
    We update its tree representation in $O(\log n)$ time. 
    
    We traverse the walk of the arrangement, maintaining the bit string and its tree representation. Whenever we arrive at some vertex $g$ (defined by the lines supporting the $i$\ts{th} edge $e_i$ and $j$\ts{th} edge $e_j$ of $P$), we compute the at most four contiguous guards that $g$ contributes to $\snapGuards$. 
    In particular, for any endpoint $x$ of a defining edge of $g$ we test in $O(\log n)$ time whether $g$ can see $x$ by testing if the shortest path data structure returns a single edge between $g$ and $x$.    
    If so, we traverse our tree representation of the bit-string corresponding to $g$  in $O(\log n)$ time to identify the set of edges $E_{g, x}$ from Definition~\ref{def:interval}. 
    We apply Lemma~\ref{lem:realComputePrimitive} to then compute in $O(\log n)$ time the maximal contiguous chain $[u, v]$ of $E_{g, x}$ that 
    is visible to $g$.
    By Observation~\ref{obs:leftOfEdges} this is the unique maximal chain containing $x$ that is visible to $g$, and we add $(g, [u, v])$ to $\snapGuards$. 
    Our space usage is dominated by storing $\snapGuards$.
\end{proof}

\begin{corollary}\label{cor:computeStart}
    We can compute the set $\candidateStarts$ from Definition~\ref{def:V^*} in $O(n^2 \log n)$ time using $O(n^2)$ space. 
\end{corollary}

\subsection{Computing a greedy maximal guard, given a $u\in\partial P$. }

We show that we can leverage our shortest-path data structure together with $\snapGuards$ to compute a greedy guard in polylogarithmic time.

\begin{lemma}\label{lem:realCompute}
    We can store a simple polygon $P$ with $n$ vertices in a data structure in $O(n^2\log n)$ time and $O(n^2)$ space, such that for any query point $u\in \partial P$ one can compute the maximal $v\in\partial P$ and point $g$ such that $(g,[u,v])$ is a contiguous guard in $O(\log^2 n)$ time.
\end{lemma}

\begin{proof}
    We preprocess $P$ through Lemma~\ref{lem:shortestPath} and \ref{lem:rayShooting}.
    We store for every $i$ and every $j\leq\lfloor\log n\rfloor$ the visibility core defined by the $i$\ts{th} to $(i+2^j)$\ts{th} edge of $P$. This requires $O(n^2)$ space and $O(n^2\log n)$ time.
    Finally, we invoke Lemma~\ref{lem:computeAllVisibility} to compute $\snapGuards$. 
    The set $\snapGuards$ induces a collection of $O(n^2)$ chains that we sort by their start vertices along $\partial P$. 
    We linearly scan these chains, removing any chains that are contained by any of its predecessors. Let $I$ be the resulting sorted set of chains. Then for given $u$ we can determine the chain $[u^-,v]\in I$, such that $u\in [u^-,v]$ and $v$ is maximal in $O(\log n)$ time.
    All these precomputation steps require $O(n^2 \log n)$ time and space.

    Given a query point $u$, we first query our chain data structure in $O(\log n)$ time to find the guard $(c, [u', v']) \in \snapGuards$ where $u \in [u', v']$ and $v'$ is maximal. We remember $v'$ and search for a better alternative.

    We binary search for the maximal edge $e$ along $\partial P$ that contains a point $v\geq v'$ where there exists a contiguous bad guard $(g, [u, v])$. 
    Observe that for a bad guard $(g, [u, v])$, the shortest path from $u$ to $v$ must be convex and only left turning.  
    Firstly, we use our shortest path data structure to compute the shortest path $\pi$ from the maximal endpoint of $e$ to $u$. 
    The data structure from~\cite{Guibas1987Optimal} can be easily adapted to test whether the output path is entirely left turning in $O(\log n)$ time (e.g.,~\cite[Observation 1]{Eades2020Visibility}). 
    
    If $\pi$ is not, then neither is the shortest path for any $v' \in e$. 
    Then, by Observation~\ref{obs:straight_is_convex}, there exists no guard that can see both $u$ and a point $v' \in e$.
    Since there is no contiguous guard that can guard $[u, v']$ for $v' \in e$, there also exists no contiguous guard that can guard from $u$ to a point on an edge succeeding $e$ and so our binary search continues by setting $e$ to an edge that is less far along $\partial P$. 
    
    If $\pi$ is convex then we compute the supporting line $\ell_u$ of the last edge of $\pi$ in $O(\log n)$ time using our shortest path data structure.
    Let $u$ lie on an edge $e_i$ and denote by $E$ the chain edges from $e_i$ to~$e$. 
    Recall that for edges $e'$ of $P$, we stored for all contiguous chains $E'$ along $\partial P$ that start or end at $e'$ whose length is a power of two its visibility core. 
    It follows that we can get two visibility cores $P_{E_1}$ and $P_{E_2}$ such that $E = E_1 \cap E_2$ in logarithmic time. 
    Since $E_1$ and $E_2$ are convex, we find in logarithmic time the intersection points between $\ell_u$ and $P_{E_1}$, and $\ell_u$ and $P_{E_2}$. 
    For these two points, we do a membership query for $P_{E_1}$ and $P_{E_2}$ to identify the first and last points along $\ell_u$ in $P_E$ (if such a point exists).

    For such a point $g$, we apply Lemma~\ref{lem:realComputePrimitive} to determine in $O(\log n)$ time the guard $(g, [u^*, v^*])$ that contains an endpoint of $e$ such that $[u^*, v^*] \subseteq E$ is maximal. Observe, that $u\in[u^*,v^*]$. 
    If $v^*$ is the maximum point seen so far, we record $(g, [u, v^*])$ and we continue the binary search by setting $e$ to an edge that is further along $\partial P$. 
    Otherwise, we we set $e$ to an edge that is less far along $\partial P$. 

    The binary search procedure terminates after $O(\log n)$ rounds, each taking $O(\log n)$ time. 
    We compare $v'$ to $v^*$ and output either $(c, [u', v'])$ or $(g, [u^*, v^*])$ depending on whether $v'$ or $v^*$ is larger. 

    To prove correctness, suppose $(g,[u'' ,v''])$ is the contiguous guard where  $u \in [u'', v'']$ that maximizes~$v''$. 
    If $(g,[u'' ,v''])$ is a good guard, then by Corollary~\ref{cor:good_is_good}, a point from $C_P$ defined by two edges intersecting $(u,v)$ also sees $[u,v]$. But, then there is an interval $[u^-,v^+]\in I$ such that $[u,v]\subseteq[u^-,v^+]$, and in particular, $v' \geq v$.
    If instead $(g,[u'',v''])$ is bad, then there are two cases by Lemma~\ref{lem:AlmostSnapBadGuards}.
    For case one, a point from $C_P$ defined by two edges intersecting $(u,v)$ also sees $[u,v]$. In this case, we may again conclude that $v' \geq v$. 
    For case two, the computed guard $(g', [u^*, v^*])$ 
    sees $[u,v]$ by Lemma~\ref{lem:AlmostSnapBadGuards}. By Lemma~\ref{lem:computeVisibility}, $[u^*, v^*]$ is the unique maximal visible chain containing $u$ from $g'$. 
\end{proof}

\section{An $O(k n^2\log^2n)$-time algorithm.}\label{sec:allTogether}

Our main result is, in spirit, an immediate consequence of Theorem~\ref{thm:good_seed} and Lemma~\ref{lem:realCompute}.

\begin{theorem}
    There is an algorithm which, provided with a simple polygon $P$ with complexity $n$, computes a set of $k$ contiguous guards $G$ such that $G$ guards the entirety of $\partial P$ in time $O(kn^2\log^2 n)$ using $O(n^2)$ space, where $k$ is the smallest number of contiguous guards necessary to guard all of $\partial P$.
\end{theorem}

\begin{proof}
    Let $\Gopt$ be an optimal solution to the \contAG problem and let $|\Gopt| = k$. If $k = 1$ then, by Observation~\ref{obs:trivialOneCenter}, the visibility core of $P$ is non-empty, which can be checked in $O(n\log n)$ time.
    Otherwise, we may assume by Theorem~\ref{thm:good_seed} without loss of generality that $\Gopt$ contains a contiguous guard $(g, [u, v])$ with $u \in \candidateStarts$ (see also Definition~\ref{def:V^*}). 

    In $O(n^2 \log n)$ time, we compute $\candidateStarts$ and the data structure of Lemma~\ref{lem:realCompute}. 
    For all $u \in \candidateStarts$, we compute a guarding solution $G_u$ by a greedy algorithm that iteratively queries for $u$ the farthest  point $v$ along $\partial P$ such that there exists a contiguous guard $(g, [u, v])$. 
    We add $(g, [u, v])$ to $G_u$ and recurse on $v$ until we find a set that guards $\partial P$. 
    By Lemma~\ref{lem:realCompute}, this procedure takes $O( |G_u| \log^2 n)$ time. 
    We output the set $G_u$ for which $|G_u|$ is minimal.

    What remains is to show that for all $u \in \candidateStarts$, $|G_u|=O(k)$, and for at least one $u \in \candidateStarts$, $|G_u| = k$. 
    If for a vertex $u \in \candidateStarts$, the optimal solution contains a contiguous guard $(g, [u, v])$ then we suppose for the sake of contradiction that $|G_u| > |\Gopt|$. 
    We order the guards in $G_u$ and $\Gopt$ by their chains along $\partial P$ from $u$.
    Since $|G_u| > |\Gopt|$, there must exist a minimum index $i$ such that for the $i$\ts{th} guards $(g_i, [u_i, v_i])\in G_u$ and $(g_i^*, [u_i^*, v_i^*])\in \Gopt$, it holds that $v_i < v_i^*$. 
    However, since $i$ is minimum, it must be that $v_{i-1} \geq u_i^*$ (if $i = 1$ then we set $v_{i-1}$ to $u$). 
    But the greedy algorithm would then have selected a contiguous guard $(g_i', [u_i', v_i'])$ with $v_i^* \leq v_i'$---a contradiction. 
    The same argument shows that for any $u\in\partial P$, and, in particular, $u\in \candidateStarts$, it holds that $|G_u| \leq k + 1$ by ordering $G_u$ and $\Gopt$ by their chains along $\partial P$, starting from the first interval that contains  $u$, but counting the guards in $\Gopt$ from $0$ instead of $1$. 

    Since preprocessing takes $O(n^2 \log n)$ time and $O(n^2)$ space, and the subsequent strategy invokes the greedy algorithm $O(n^2)$ times, and each invocation taking $O(k \log^2 n)$ time, the theorem follows. 
\end{proof}
\bibliography{bibliography}

\appendix






\end{document}